\documentclass{lmcs} 


\usepackage{hyperref}

\newcommand{\bi}{\begin{itemize}}
	\newcommand{\ei}{\end{itemize}}

\newcommand{\be}[1]{\begin{equation}\label{#1}}
	\newcommand{\ee}{\end{equation}}

\newcommand{\A}{\mathcal{A}}
\newcommand{\R}{\mathbb{R}}
\newcommand{\M}{\mathcal{M}}
\renewcommand{\S}{\mathcal{S}}
\newcommand{\G}{\mathcal{G}}

\newcommand{\PP}{\mathbb{P}}

\newcommand{\emptystr}{\epsilon}

\DeclareMathOperator{\nbocc}{NbOcc}

\DeclareMathOperator{\freq}{Freq}
\DeclareMathOperator{\bfreq}{BFreq}

\DeclareMathOperator{\select}{Select}
\DeclareMathOperator{\capital}{Capital}

\theoremstyle{plain}

\newtheorem{theorem}{Theorem}[section]
\newtheorem*{theorem*}{Theorem}

\newtheorem{lemma}[theorem]{Lemma}
\newtheorem*{lemma*}{Lemma}

\newtheorem*{proposition*}{Proposition}

\newtheorem*{corollary*}{Corollary}

\newtheorem*{claim*}{Claim}

\newtheorem*{observation*}{Observation}

\newtheorem*{conjecture*}{Conjecture}

\newtheorem{definition}[theorem]{Definition}
\newtheorem*{definition*}{Definition}

\newtheorem*{example*}{Example}

\newtheorem*{problem*}{Problem}

\newtheorem*{question*}{Question}

\newtheorem*{assumption*}{Assumption}

\newtheorem*{notation*}{Notation}

\newtheorem*{remark*}{Remark}

\newtheorem*{convention*}{Convention}

\begin{document}
		
		\title[Agafonov and Schnorr-Stimm for probabilistic automata]
		{The Agafonov and Schnorr-Stimm theorems for probabilistic automata}
		
		\thanks{Laurent Bienvenu and Subin Pulari were supported in part by the ANR project FLITTLA ANR-21-CE48-0023. Subin Pulari was supported by the Basic Research Program at HSE University and this paper was prepared in the framework of this program.}
		
		\author[L.~Bienvenu]{Laurent Bienvenu\lmcsorcid{0000-0002-9638-3362}}[a]
		\author[H.~Gimbert]{Hugo Gimbert\lmcsorcid{0000-0003-1227-9718}}[a]
		\author[S.~Pulari]{Subin Pulari\lmcsorcid{0000-0001-8426-4326}}[b]
		
		\address{LaBRI, CNRS \& Université de Bordeaux, France}
		\email{laurent.bienvenu@labri.fr, hugo.gimbert@labri.fr}
		
		\address{National Research University Higher School of Economics Moscow}
		\email{spulari@hse.ru}
		
		

\begin{abstract}
  \noindent For a fixed alphabet~$A$, an infinite sequence~$X$ is said to be normal if every word $w$ over~$A$ appears in~$X$ with the same frequency as any other word of the same length. A classical result of Agafonov~(1968) relates normality to finite automata as follows: a sequence~$X$ is normal if and only if any subsequence of~$X$ selected by a finite automaton is itself normal. Another theorem of Schnorr and Stimm~(1972) gives an alternative characterization: a sequence~$X$ is normal if and only if no gambler can win large amounts of money by betting on the sequence~$X$ using a strategy that can be described by a finite automaton. Both of these theorems are established in the setting of deterministic finite automata. This raises the question as to whether they can be extended to the setting of probabilistic finite automata. In the case of the Agafonov theorem, a partial positive answer was given by L\'echine et al.\ (2024) in the restricted case of probabilistic automata with rational transition probabilities. 
  
  In this paper, we settle the full conjecture by proving that both the Agafonov and the Schnorr-Stimm theorems hold true for arbitrary probabilistic automata. Specifically, we show that a sequence~$X$ is normal if and only if any probabilistic automaton selects a normal subsequence of~$X$ with probability~$1$ and also show that a sequence~$X$ is normal if and only if any probabilistic finite-state gambler fails to win on~$X$ with probability $1$. 
\end{abstract}

\maketitle

\section{Introduction}

Given a finite alphabet~$A$ of $k$ letters, an infinite sequence~$X$ of letters is said to be \emph{normal} if every word~$w$ of $A^*$ appear as subword of~$X$ with the same frequency as any other word of the same length, namely, $(1/k)^{|w|}$. The famous Champernowne sequence 
\[
01234567891011121314151617181920212223242526\ldots
\]
can be shown to be normal over the alphabet $A=\{0, \ldots, 9\}$. The number~$\pi$, for example, is conjectured to have a normal expansion in every base, though this very much remains an open question. Normal sequences are plentiful, and an easy way to generate a normal sequence~$X$ is to draw each letter $X(n)$ at random in the alphabet $A$ (all letters having the same probability $1 / |A|$) independently of the other chosen letters $X(m)$. The law of large numbers tells us that we obtain a normal sequence with probability~$1$. \\

Of course, there are also plenty of examples of non-normal sequences: \\
\begin{itemize}
	\item Periodic, or ultimately periodic, sequences can never be normal: indeed if the period of~$X$ is~$k$ there are only $k$ possible subwords of~$X$ of length~$k$ hence most words of length~$k$ will not appear and their frequency will be~$0$. 
	\item Sturmian sequences, which are sequences with only $k+1$ different subwords of length~$k$ (such as the Fibonacci sequence 100101001001010010100100101001001... obtained by iterating the morphism $0 \mapsto 01$ and $1 \mapsto 0$) are not normal for the same reason. 
	\item The Thue-Morse sequence $01101001100101101001011001101001...$ obtained by iterating the morphism $0 \mapsto 01$ and $1 \mapsto 10$ is not normal because $000$ and $111$ do not appear as subwords. 
	\item A sequence of $0$'s and $1$'s generated at random where each bit is chosen equal to the previous one with probability $2/3$ will have, with probability~$1$, all possible finite words as subwords but will not (still with probability $1$) be normal as for example the word $00$ will appear with frequency $1/3$ instead of $1/4$. 
\end{itemize}

It turns out that normality has a nice interpretation in terms of finite automata. Indeed, two classical results, one due to Agafonov \cite{Agafonov1968} and the other due to Schnorr and Stimm \cite{SchnorrS1972}, assert that an infinite sequence is normal if and only if it cannot be predicted by a finite automaton with better-than-average accuracy. Of course, one needs to specify what `predicted' means. We consider two prediction models. 

\begin{itemize}
	\item[(I)] In the Agafonov model, an automaton reads the infinite sequence~$X$ one letter at a time and updates its state in the usual way. Some of its states have a `select' tag on them. When the current state has such a tag, the next letter will be selected and added to a subsequence~$Y$. We consider the automaton successful at predicting~$X$ if the subsequence~$Y$ built in this process is infinite and some letter of $A$ does not have asymptotic frequency $1/|A|$ in~$Y$. This means that the automaton has exhibited a statistical anomaly in the sequence~$X$ and isolated this anomaly in the subsequence~$Y$. 
	
	\item[(II)] In the Schnorr-Stimm model, the predictor is still an automaton but this time is viewed as a gambling strategy. The gambler starts with a capital of $\$1$. Each state~$q$ is labeled with a betting function $\gamma_q: A \rightarrow \R^{\geq 0}$. This function represents the amount by which the predictor would like her capital to be multiplied by depending on the value of the next bit. For example, suppose the player plays against a sequence $X \in \{a,b,c\}^\omega$. If her current capital is $\$2$ and  the current state~$q$ is labelled by a betting function $\gamma_q$ such that $\gamma_q(a)=0.7$, $\gamma_q(b)=1.1$ and $\gamma(c)=1.2$, if the next letter is $a$, her new capital will be $\$1.4$, if it is $b$ her new capital will be $\$ 2.2$ and if it is~$c$ her new capital will be $\$2.4$. For the game to be fair, each betting function $\gamma_q$ must have expectation equal to~$1$, i.e., must satisfy $\frac{1}{|A|}\sum_{a \in A} \gamma_q(a) = 1$. We say that the predictor wins if the capital of the player takes arbitrarily large values throughout the (infinite) game. That is, the predictor has spotted some type of statistical anomaly and is exploiting it to get rich!
\end{itemize}

Both of these models lead to the same conclusion: an infinite sequence $X$ is normal if and only if it is unpredictable (by a finite automaton). 

\begin{theorem}[Agafonov~\cite{Agafonov1968}]
	For $X \in A^\omega$, the following are equivalent.
	\item[(i)] $X$ is normal. 
	\item[(ii)] For any automaton $\A$ that selects a subsequence~$Y$ of~$X$ as in model~I, either~$Y$ is finite, or every letter of $A$ appears in~$Y$ with asymptotic frequency $1/|A|$.
	\item[(iii)] For any automaton $\A$ that selects a subsequence~$Y$ of~$X$ as in model~I, either~$Y$ is finite, or $Y$ is normal. 
\end{theorem}

(See also Carton~\cite{Carton2020} and Seiller and Simonsen~\cite{SeillerS2020} for a modern account of the above theorem).

\begin{theorem}[Schnorr-Stimm~\cite{SchnorrS1972}]
	For $X \in A^\omega$, the following are equivalent.
	\item[(i)] $X$ is normal. 
	\item[(ii)] Any automaton $\A$ betting on~$X$ according to model~II does not win. 
\end{theorem}

In both of these theorems, the finite automata used for prediction are assumed to be deterministic. Would the situation change if one allowed probabilistic automata? In principle, one would not expect an unpredictable sequence to become predictable in the presence of a random source. Indeed, given a sequence~$X$ and a random source~$R$ it seems, informally speaking, that almost surely~$R$ will not `know' anything about~$X$ and thus will not help in predicting~$X$. Surprisingly, this intuition is wrong in the setting where the predictors are not finite automata but are Turing machines, as shown by Bienvenu et al.~\cite{BienvenuDS2022} who built a sequence that is unpredictable by deterministic Turing machines (in either prediction model of selection or gambling) and becomes predictable (in either model) if one allows probabilistic Turing machines. Nonetheless, finite automata are much weaker than Turing machines and Bienvenu et al.'s construction does not work for such a memoryless model of computation. Indeed, recently, Léchine et al.\ showed that Agafonov's theorems holds for probabilistic automata in the restricted case where the transition probabilities are rational. 

\begin{theorem}[Léchine et al.~\cite{LechineSS2024}]
	For $X \in A^\omega$, the following are equivalent.
	\item[(i)] $X$ is normal. 
	\item[(ii)] For any probabilistic automaton $\A$ with rational probabilities, almost surely, $\A$ selects a subsequence~$Y$ of~$X$ that is either finite, or every letter of $A$ appears in~$Y$ with asymptotic frequency $1/|A|$.
	\item[(iii)] For any probabilistic automaton $\A$ with rational probabilities, almost surely, $\A$ selects a subsequence~$Y$ of~$X$ such that either~$Y$ is finite, or $Y$ is normal. 
\end{theorem}

This led them to conjecture that the probabilistic version of Agafonov's theorem holds in the general case, where the probabilities are real-valued rather than being rational. In this paper, we prove this conjecture and also prove the probabilistic version of the Schnorr-Stimm theorem. Additionally, we establish a probabilistic version of the Schnorr-Stimm dichotomy regarding the winning rates of probabilistic gamblers.

Léchine et al.'s proof is a reduction of the rational probabilistic case to the deterministic case by some clever combinatorial reduction (which substantially increases the number of states). We propose two different proofs of the general theorem (removing the `rational' assumption). Our first one is also a reduction to the deterministic case but is of very different nature: Instead of encoding a probabilistic automaton into a deterministic one, we will use the fact that a probabilistic automaton can be seen as a deterministic one over a larger alphabet, without altering the number of states. For this, we will need an extension of the deterministic case to Bernoulli measures (an extension which was proved by Seiller and Simonsen~\cite{SeillerS2020}), which will be presented in the next section. Our second proof is more direct (at the price of being more technical), i.e., it does analyze the probabilistic process coming out of the run of a probabilistic automaton on a given sequence and argues that this process behaves `as expected' when~$X$ is normal. We believe that both proofs are informative which is why we chose to include both.\\

\subsection*{Notation and terminology}

We finish this introduction by formalizing the concepts discussed so far and gathering the notation and terminology that will be used in the rest of the paper. \\

Given an alphabet $A$, we denote by $A^*$ the set of finite words over $A$, by $A^n$ the set of words of length~$n$, by $A^\omega$ the set of infinite sequences of letters and by $A^{\leq \omega}$ the set $A^* \cup A^\omega$. For $X \in A^{\leq \omega}$, we denote by $X(i)$ the $i$-th letter of $X$ (by convention there is a $0$-th letter) and by $X[i,j]$ the word $X(i)X(i+1)\ldots X(j)$. Let $\emptystr$ denote the empty word.\\

Given a word~$u$ of length~$k$ and a word~$w$ of length $n \geq k$, we denote by $\nbocc(u,w)$ the number of occurrences of the word~$u$ in~$w$, i.e.,
\[
\nbocc(u,w) = \# \{i : 0 \leq i \leq n-k, \ w[i,i+k-1]=u\}
\]
and the frequency of occurrence $\freq(u,w)$ of~$u$ in~$w$ is naturally defined by 
\[
\freq(u,w) = \frac{\nbocc(u,w)}{n-k+1}
\]
When~$X$ is an infinite sequence, we define
\[
\freq^-(u,X) = \liminf_{n \rightarrow \infty}\, \freq(u,X[0,n]) ~ ~ \text{and} ~ ~ \freq^+(u,X) = \limsup_{n \rightarrow \infty}\,  \freq(u,X[0,n])
\]
When $\freq^-(u,X)$ and $\freq^+(u,X)$ have the same value, we simply call this common value $\freq(u,X)$. \\

Given $X \in A^\omega$, we say that $X$ is \emph{balanced} if all letters appear in~$X$ with the expected frequency, i.e., $\freq(a,X)=1/|A|$ for all~$a \in A$. We say that $X$ is \emph{normal} if all words appear in $X$ with the expected frequency, i.e., $\freq(w,X)=|A|^{-|w|}$ for all $w \in A^*$. \\

A \emph{deterministic finite automaton (DFA)} is a tuple $(Q,A,q_I,\delta)$ where $Q$ is a finite set of states, $A$ a finite alphabet, $q_I$ the initial state and $\delta: Q \times 	A \rightarrow Q$ the transition function (in this paper, runs of automata are meant to be infinite hence there is no need for final states). We denote by $\delta^*$ the function from $Q \times A^*$ to $Q$ defined inductively by $\delta^*(q,\emptystr)=q$ and for $w \in A^*$ and $a \in A$, $\delta^*(q,w \cdot a) = \delta(\delta^*(q,w),a)$, where $\cdot$ is the concatenation of words. 

An \emph{automatic selector} (or \emph{selector} for short) is a tuple $(Q,A,q_I,\delta,S)$ where $(Q,A,q_I,\delta)$ is a DFA and $S$ is a subset of $Q$, representing the selection states. 

Given a selector $\S=(Q,A,q_I,\delta,S)$, we define the selection function from $A^*$ to $A^*$ inductively by:
\[
\select(\S,\emptystr)=\emptystr
\]
and for $w \in A^*$ and $a \in A$:
\[
\select(\S,w \cdot a) = \left \{
\begin{array}{lr}
	\select(\S,w) \cdot a &  \text{ if } \delta^*(q_I,w) \in S\\
	\select(\S,w) &  \text{ if } \delta^*(q_I,w) \notin S\\
\end{array}
\right.
\]

If $X$ is an infinite sequence in $A^\omega$, the sequence of words $\select(\S,X[0,n])$ is non-decreasing with respect to the prefix order and thus converges to a sequence $Y \in A^{\leq \omega}$ which we call \emph{the selected subsequence of $X$ selected by $\S$} and denote by $\select(\S,X)$. \\

An \emph{automatic gambler} (or \emph{gambler} for short) is a tuple $(Q,A,q_I,\delta,\gamma)$ where $(Q,A,q_I,\delta)$ is a DFA and $\gamma$ is a function from $Q \times A$ to $\R^{\geq 0}$ such that for all~$q$, $\frac{1}{|A|}\sum_{a \in A} \gamma(q,a)=1$. As said earlier, the value of $\gamma(q,a)$ should be interpreted as the multiplier by which the gambler, being currently in state~$q$, would like her capital to be multiplied by if the next read letter is~$a$. The condition $\frac{1}{|A|} \sum_{a \in A} \gamma(q,a)=1$ ensures that the game is fair. \\

Given a gambler $\G=(Q,A,q_I,\delta,\gamma)$ and $w \in A^*$, we define $\capital(\G,w)$ inductively by
\[
\capital(\G,\emptystr)=1
\] 
and for $w \in A^*$ and $a \in A$,
\[
\capital(\G,w \cdot a)=\capital(\G,w) \cdot \gamma(\delta^*(q_I,w),a)
\]
and we say that a gambler $\G$ \emph{wins against $X \in A^\omega$} if 
\[
\limsup_{n \rightarrow +\infty}\,  \capital(\G,X[0,n]) = +\infty
\]
(otherwise we say that $\G$ loses).

\section{Deterministic prediction for Bernoulli measures}

In classical normality, all letters of the alphabet occur with the same frequency. We can however consider the extension of normality to Bernoulli measures. A Bernoulli measure over $A^\omega$ is a probability measure where letters of an infinite sequence~$X$ are drawn at random independently of one another but the distribution over the alphabet $A$ is non-uniform. 

\begin{definition}
	Let  $\mu: A \rightarrow [0,1]$ be a distribution over the alphabet $A$ (hence satisfies $\sum_{a \in A} \mu(a) = 1$). The Bernoulli measure induced by~$\mu$, which we also denote by $\mu$ by abuse of notation, is the unique probability measure such that for all~$i,k$, for every word $w=a_0, \ldots a_k \in A$,
	\[
	\Pr_{X \sim \mu} \Big[ X[i,i+k]=w \Big]= \prod_{j=0}^k \mu(a_{j})
	\]
	We also denote by $\mu(w)$ the quantity $\prod_{j=0}^k \mu(a_{j})$.    
\end{definition}

Normality generalizes very naturally to Bernoulli measures.

\begin{definition}
	Let $\mu$ be a Bernoulli measure. A sequence~$X \in A^\omega$ is \emph{$\mu$-balanced} if $\freq(a,X)=\mu(a)$ for all~$a \in A$. It is \emph{$\mu$-normal} if for all words~$w \in A^*$, $\freq(w,X)=\mu(w)$.
\end{definition}

We say that a Bernoulli measure~$\mu$ is \emph{positive} when $\mu(a)>0$ for every letter~$a$. In the rest of the paper, all Bernoulli measures will be assumed to be positive, and we simply say `Bernoulli measure' to mean `positive Bernoulli measure'.

The Agafonov theorem can be extended to Bernoulli measures, as proven by Seiller and Simonsen~\cite{SeillerS2020}. It is this theorem that we will use in the next section to obtain a proof of the Agafonov theorem for probabilistic selectors.

\begin{theorem}[Agafonov theorem for Bernoulli measures~\cite{SeillerS2020}] \label{thm:agafonov-bernoulli}
	Let~$\mu$ be a Bernoulli measure. For $X \in A^\omega$, the following are equivalent.
	\item[(i)] $X$ is $\mu$-normal. 
	\item[(ii)] For any selector $\S$ that selects a subsequence~$Y$ of~$X$, either~$Y$ is finite or $Y$ is $\mu$-balanced.
	\item[(iii)] For any selector $\S$ that selects a subsequence~$Y$ of~$X$, either~$Y$ is finite or $Y$ is $\mu$-normal.
\end{theorem}

We can also easily generalize the notion of gambler to the setting of Bernoulli measures: it suffices to define a $\mu$-gambler $\G=(Q,A,q_I,\delta,\gamma)$ as before but with the fairness condition on~$\gamma$ replaced by $\sum_{a \in A} \mu(a)\gamma(q,a) =1$ for every~$q$. The function $\capital$ and the notion of success are defined as before. \\ 

We will now prove that the Schnorr-Stimm theorem, just like the Agafonov theorem, can also be extended to Bernoulli measures.

\begin{theorem}[Schnorr-Stimm theorem for Bernoulli measures] \label{thm:schnorr-stimm-bernoulli}
	For $X \in A^\omega$ and a Bernoulli measure $\mu$, the following are equivalent.
	\item[(i)] $X$ is $\mu$-normal. 
	\item[(ii)] No $\mu$-gambler $\G$ wins by betting on~$X$.  
\end{theorem}

\begin{proof}
	$(i) \Rightarrow (ii)$. Suppose that $X \in A^\omega$ is $\mu$-normal and consider a $\mu$-gambler $\G=(Q,A,q_I,\delta,\gamma)$. 
	
	We can assume that the $\mu$-gambler only has one state on which it places a non-trivial bet. Indeed, define for every $q \in Q$ the $\mu$-gambler $\G^{[q]}=(Q,A,q_I,\delta,\gamma^{[q]})$ where $\gamma^{[q]}(q,a)=\gamma(q,a)$ for all~$a$ and $\gamma^{[q]}(q',a)=1$ for $q' \not = q$. That is, $\G^{[q]}$ is the gambler $\G$ where all states but state~$q$ are neutralized (no bet is placed while on them). By the multiplicative nature of the function $\capital$, we have for all~$n$:
	\[
	\capital(\G,X[0,n]) = \prod_{q \in Q} \capital(\G^{[q]},X[0,n])
	\]
	Thus if we can show that all quantities $\capital(\G^{[q]},X[0,n])$ are bounded, we are done. Let us assume that there is a state~$r$ that is the unique state on which $\G$ bets. If instead of $\G$ we consider the selector $\S=(Q,A,q_I,\delta,\{r\})$ with only~$r$ as the selecting state, we know by Agafonov's theorem for Bernoulli measures (Theorem~\ref{thm:agafonov-bernoulli}) that the subsequence $Y$ of $X$ selected by $\S$ is $\mu$-normal, hence in particular $\mu$-balanced. But this subsequence is precisely the values of $X$ on which $\G$ bets! 
	
	We can further assume that $Y$ is infinite, otherwise it means that the run of $\G$ on~$X$ passes by~$r$ finitely often, hence $\G$ certainly cannot win as other states are not betting states. Now, suppose that at the $n$-th step of the run on~$X$ the state~$r$ has been visited $k=k(n)$ times. We have 
	\[
	\capital(\G,X[0,n]) = \prod_{a \in A} \gamma(r,a)^{\nbocc(a,Y[0,k])}
	\] 
	But since $Y$ is $\mu$-normal we have, for every~$a$, $\nbocc(a,Y[0,k]) = \mu(a)k + o(k)$. Thus,
	\[
	\capital(\G,X[0,n]) = \prod_{a \in A} \gamma(r,a)^{\mu(a)k + o(k)}
	\] 
	or equivalently 
	\[
	\log \capital(\G,X[0,n]) = (k+o(k)) \cdot \sum_{a \in A} \mu(a)\cdot  \log \gamma(r,a) 
	\] 
	(here we assume that all values $\gamma(r,a)$ involved in the product are positive for if not then the capital falls to~$0$ and we are done). Since we have  $\sum_{a \in A} \mu(a)=1$ ($\mu$ being a distribution), we can use the strict concavity of the function $\log$ on $(0,+\infty)$ to apply Jensen's inequality and get 
	\[
	\sum_{a \in A} \mu(a)\cdot  \log \gamma(r,a) \leq \log \left( \sum_{a \in A} \mu(a) \gamma(r,a)\right)
	\]
	with strict inequality when not all $\gamma(q,a)$ are equal (which is the case where $\G$ makes non-trivial bets). But by the fairness condition, we have $\sum_{a \in A} \mu(a) \gamma(r,a)=1$ hence we see that $\log \capital(\G,X[0,n])$ is either $0$ or ultimately negative which either way means that $\G$ does not win. \\
	
	$(ii) \Rightarrow (i)$. Assume that $X$ is not $\mu$-normal. This means that there is some word~$w$ such that $\freq(w,X[0,n])$ does not converge to $\mu(w)$. Let us assume that~$w$ is such a word of minimal length and write $w=ux$ with $u \in A^*$ and $x \in A$. 
	
	Consider the sequence of vectors $f_n$ defined by
	\[
	f_n = \left( \frac{\freq(ua,X[0,n])}{\sum_{b \in A} \freq(ub,X[0,n])} \right)_{a \in A}
	\]
	
	or equivalently, 
	\[
	f_n = \left( \frac{\nbocc(ua,X[0,n])}{\sum_{b \in A} \nbocc(ub,X[0,n])} \right)_{a \in A} 
	\]

	All of these vectors belong to the set $\Gamma=\{f: A \rightarrow [0,1] : \sum_{a \in A} f(a)=1\}$. This is a compact set, hence the sequence $f_n$ must have cluster points. By definition of $u$ and $x$, we know that $f_n$ does not converge to $\mu$ because $f_n(x)$ does not converge to $\mu(x)$: Indeed, in the definition of $f_n(x)$, the denominator converges to $\freq(u,X)$ which by minimality of~$w$ is defined and equal to~$\mu(u)$, while the numerator is equal to $\freq(w,X[0,n])$ which by definition of $w$ does not converge to $\mu(w)=\mu(u)\mu(x)$.

	Therefore, the sequence $(f_n)$ must have at least one cluster point~$\nu$ different from $\mu$. Fix such a cluster point $\nu$. 
	
	We now build our gambler $\G=(Q,A,q_I,\delta,\gamma)$. The idea is that the gambler will record the last $|u|$ bits it read and will only place bets when these exactly form the word~$u$. Let us thus take $Q=\{q_v : v \in A^*, |v| \leq |u|\}$ initial state $q_I=q_\emptystr$ and define $\delta$ by 
	\begin{align*}
		\delta(q_v,a)= \begin{cases}
			q_{va}, & \text{if } \lvert v \rvert < |u|,\\
			q_{v'a} & \text{ if } |v|=|u| \text{ and } v=xv' \text{ with } x \in A.		\end{cases}
	\end{align*}
	
	Now, define $\gamma(q_v,a)=1$ whenever $v \not= u$ and $\gamma(q_u,a) = \nu(a)/\mu(a)$ for all~$a \in A$. Observe that this is a valid $\mu$-gambler as the fairness condition is satisfied: $\sum_{a \in A} \nu(a)/\mu(a) \cdot \mu(a) = \sum_a \nu(a) = 1$.
	
	Suppose that after reading $n$ letters of~$X$ the state $q_u$ has been visited~$k=k(n)$ times. First, observe that $k(n)$ tends to $+\infty$. Indeed, the state $q_u$ is visited whenever $u$ is seen as a subword of $X$. But we assumed that $u$ appears in~$X$ with frequency $\mu(u)$, by minimality of~$w$. 
	
	
	Second, unfolding the definition of $f_n$, we have,
	\[
	\capital(\G,X[0,n]) = \prod_{a \in A} \left( \frac{\nu(a)}{\mu(a)} \right)^{k \cdot f_n(a)}.
	\] 
	This gives
	\[
	\log \capital(\G,X[0,n])= k \cdot \sum_{a \in A} f_n(a) \log \left( \frac{\nu(a)}{\mu(a)} \right).
	\]
	Since $\nu$ is a cluster point of the sequence $f_n$, for any fixed $\varepsilon>0$ there are infinitely many~$n$ (or $k$) such that
	\[
	\sum_{a \in A} f_n(a) \log \left( \frac{\nu(a)}{\mu(a)} \right) \geq \sum_{a \in A} \nu(a) \log \left( \frac{\nu(a)}{\mu(a)} \right)  - \varepsilon.
	\]
	But the term $\sum_{a \in A} \nu(a) \log \left( \frac{\nu(a)}{\mu(a)} \right)$ is the relative entropy from $\mu$ to $\nu$ (also known as Kullback-Liebler divergence, see for example~\cite{CoverT2006}) which we denote by $\mathrm{D}_{KL}(\nu || \mu)$. This quantity is non-negative in general and is positive when $\nu \not= \mu$, which is the case here. We have thus established that for any fixed $\varepsilon$, there are arbitrarily large~$n$ and $k$ such that 
	\[
	\log \capital(\G,X[0,n]) \geq k \left( \mathrm{D}_{KL}(\nu || \mu) - \varepsilon \right).
	\]
	Taking any $\varepsilon < \mathrm{D}_{KL}(\nu || \mu)$, this shows that 
	\[
	\limsup_{n \rightarrow +\infty} \capital(\G,X[0,n]) = +\infty.
	\]\end{proof}

Let us note that this last proof actually gives us a finer analysis of normality, in terms of the rate of failure or success of the gambler. This was already observed by Schnorr and Stimm in their seminal paper, where they proved the following.

\begin{theorem}[Schnorr-Stimm dichotomy theorem~\cite{SchnorrS1972}]
	Let $X$ be an infinite sequence in $A^\omega$. \\
	(i) If $X$ is normal and $\G$ is a gambler, then the capital of $\G$ throughout the game either is ultimately constant or decreases at an exponential rate.\\
	(ii) If $X$ is not normal, then there exists a gambler~$\G$ which wins against~$X$ at an `infinitely often' exponential rate (i.e., $\limsup_n \log(\capital)/n >0$). 
\end{theorem}

As a byproduct of our proof of~Theorem~\ref{thm:schnorr-stimm-bernoulli}, we have the same dichotomy for positive Bernoulli measures (i.e., Bernoulli measures such that $\mu(a)>0$ for every letter):

\begin{theorem}[Schnorr-Stimm dichotomy theorem for Bernoulli measures] \label{thm:schnorr-stimm-dicho-bernoulli}
	\label{thm:bernoullidichotomy}
	Let $X$ be an infinite sequence in $A^\omega$ and $\mu$ a positive Bernoulli measure. \\
	(i) If $X$ is $\mu$-normal and $\G$ is a $\mu$-gambler, then the capital of $\G$ throughout the game either is ultimately constant or decreases at an exponential rate.\\
	(ii) If $X$ is not $\mu$-normal, then there exists a $\mu$-gambler~$\G$ which wins against~$X$ at an `infinitely often' exponential rate.
\end{theorem}

Our proof of~Theorem~\ref{thm:schnorr-stimm-bernoulli} almost establishes this, but we do need an additional technical lemma and the block-wise characterization of $\mu$-normality. Let $\bfreq(u, w)$ denote the block frequency of the word $u$ in $w$, defined as the proportion of non-overlapping blocks in $w$ that are equal to~$u$. More precisely when $n=\lvert w \rvert$ and $k=\lvert u \rvert$,
\[
\bfreq(u,w) = \frac{	\# \{i : 0 \leq i \leq \lfloor n/k \rfloor~,~  w[ki,k(i+1)-1]=u
	\}}{\lfloor n/k \rfloor}. 
\]
For $w \in A^*$, let $\bfreq^-(w,X)$, $\bfreq^+(w,X)$ and $\bfreq(w,X)$ denote the lower, upper and limit block frequency of $w$ in $X$ defined similarly as $\freq^-(w,X)$, $\freq^+(w,X)$ and $\freq(w,X)$. A sequence $X$ is \emph{$\mu$-block normal} if for all words~$w \in A^*$, $\bfreq(w,X)=\mu(w)$. The following was shown by  Seiller and Simonsen.

\begin{lemma}[Seiller-Simonsen~\cite{SeillerS2020}]
	If a sequence $X \in A^\omega$ is $\mu$-normal then $X$ is $\mu$-block normal.
\end{lemma}

We note that the converse implication  also holds.

\begin{lemma}
	\label{lem:blocknormalimpliesnormal}
	If a sequence $X \in A^\omega$ is $\mu$-block normal then~$X$ is $\mu$-normal.
\end{lemma}

The above lemmas completes the proof of the following equivalence theorem between $\mu$-normality and $\mu$-block normality.

\begin{theorem}
	\label{thm:mu-blocknormal-equivalence}
	A sequence $X \in A^\omega$ is $\mu$-normal if and only if $X$ is $\mu$-block normal.
\end{theorem}

The proof of Lemma~\ref{lem:blocknormalimpliesnormal} uses a counting trick from the proof of Theorem 3.1 from~\cite{NandakumarPVV2021} which in turn is based on the proof of the main theorem in \cite{Maxfield1952}.

\begin{proof}[Proof of Lemma \ref{lem:blocknormalimpliesnormal}]
	As in the proof of Theorem 3.1 from \cite{NandakumarPVV2021}, for any finite length string
	$w=a_1a_2a_3\dots a_k\in A^k$ and large enough $n$, 
	\begin{align}
		\label{eq:pillaistheorem_one}
		\freq(w,X[0,n])=
		f_1(n) + f_2(n) +\dots+
		f_{(1+\left \lfloor\log_2  \frac{n}{k} \right \rfloor)}(n) + {\frac{(k-1)\cdot O(\log n) }{n-k+1}}
	\end{align}
	where $f_p(n)$ are defined as follows:
	\begin{align*}
		f_{p}(n)= \begin{cases}
			\frac{	\# \{i~ : ~ X[ki,k(i+1)-1]=w
				,~ 0 ~\leq ~i ~\leq ~\lfloor n/k \rfloor \}}{n-k+1}, & \text{if } ~p=1\\
				~ ~ \\
			\sum_{j=1}^{k-1}\frac{  \# \{i~ : ~ X[2^{p-1}ki,2^{p-1}k(i+1)-1] \in S_j
				,~ 0 \leq i \leq n/2^{p-1}k \}}{n-k+1}, &\text{if } 1<p \leq
			(1+\left \lfloor\log_2(n/k)\right \rfloor)\\ 
			~ \\
			0, \text{ otherwise.}
		\end{cases}
	\end{align*}
	In the above definition, $S_{j}$ is the set of strings of the form,
	$u\ a_1a_2 \dots a_k\ v$ where $u$ is some string of length $2^{p-2}k-j$ and $v$ is some string of length $2^{p-2}k-k+j$.
	
	Equation \ref{eq:pillaistheorem_one} shows that the frequency of $w$ in $X[0,n]$ can be written as a sum of different block frequencies. The quantity $f_1(n)$ counts the number of occurrences of~$w$ inside disjoint $k$-length blocks in $X[0,n]$, $f_2(n)$ counts the number of occurrences crossing the boundaries of these $k$-blocks, $f_3(n)$ counts those crossing boundaries of $2k$-blocks, and so on. In general, $f_p(n)$ counts the number of occurrences straddling boundaries of disjoint $2^{p-1}k$-length blocks. Each $f_p(n)$ may miss at most $(k-1)$ occurrences of~$w$, which are accounted for by the error term.
	
	Since $X$ is $\mu$-block normal, 
	\begin{align*}
		\lim_{n\to\infty} f_1(n) = \lim_{n\to\infty}\frac{	\# \{i~ : ~ X[ki,k(i+1)-1]=w
			~,~ 0 ~\leq ~i ~\leq ~\lfloor n/k \rfloor \}}{n-k+1} =
		\frac{\mu(w)}{k}.
	\end{align*}
	Now, when $p \leq 1+ \left\lfloor\log_2(\frac{n}{k}) \right\rfloor$,
	\begin{align*}
		\lim_{n\to\infty} f_p(n)
		=&  \sum_{j =1}^{k-1}\lim_{n\to\infty}
		\frac{  \# \{i ~: ~  X[2^{p-1}ki,2^{p-1}k(i+1)-1] \in S_j
			, 0 \leq i \leq n/2^{p-1}k \}}{n-k+1}\\
		&= \frac{\mu(w)}{2^{p-1}k}(k-1).
	\end{align*}
	Since
	$\langle \sum_{i=1}^{m} f_i(n) \rangle_{m \in \mathbb{N}}$ is uniformly
	convergent, we have
	\begin{align*}
		\freq(w,X) = \lim_{n\to\infty} \sum_{i=1}^{\infty} f_{i}(n)=  \sum_{i=1}^{\infty} \lim_{n\to\infty} f_{i}(n).
	\end{align*}
	Therefore from (\ref{eq:pillaistheorem_one}),
	\begin{align*}
		\lim_{n\to\infty} \freq(w,X[0,n])
		&= \frac{\mu(w)}{k} + (k-1)\mu(w) \Big[ \sum_{i=1}^\infty \frac{1}{2^i k} \Big] = \mu(w).
	\end{align*}
	Hence, $X$ is $\mu$-normal.
\end{proof}

We first recall two standard facts about finite Markov chains. We state them in a form that will be used below. We view distributions as row vectors.

\begin{lemma}[Limit average distribution \cite{Norris_1997}]
	\label{lem:limit-average-distribution}
	Let $\mathcal{M}$ be a finite homogeneous Markov chain with states $S$
	and transition matrix $P \in [0,1]^{S\times S}$.
	Assume $\mathcal{M}$ is ergodic, i.e. all states are recurrent and belong to the same
	recurrence class.
	Fix some initial state $s_0\in S$.
	Define $\eta \in \Delta(S)$ by
	\[
	\eta(t)  = \lim_{N\rightarrow +\infty} \frac{1}{N}\sum_{n < N} \mathbb{P}_{s_0}\left(S_n = t\right)\enspace .
	\]
	Then $\eta$ is well-defined, the limit exists, and it is independent of the choice of $s_0$.
	Moreover, $\eta$ is the unique solution to the affine system
	\[
	\eta P = \eta
	\qquad\text{and}\qquad
	\sum_{s\in S}\eta(s)=1\enspace .
	\]
\end{lemma}

\begin{lemma}
	\label{lem:stationary-distribution-continuity}
	Let $\mathcal{M}$ be a finite homogeneous Markov chain with states $S$
	and transition matrix $P \in [0,1]^{S\times S}$.
	Assume $\mathcal{M}$ is ergodic, i.e. all states are recurrent and belong to the same
	recurrence class.
	The limit average distribution is locally continuous in $P$.
\end{lemma}

\begin{proof}
	In the affine system of Lemma~\ref{lem:limit-average-distribution},
	there are $|S|$ variables and $|S|+1$
	equations. We can remove one of the equations in order to obtain a square matrix
	so that $\eta$ is the unique solution to $R\cdot \eta = y$, since
	\[
	\eta = R^{-1}\cdot y\enspace .
	\]
	In a neighbourhood of $R$, the determinant is non-zero and the inverse is continuous,
	hence $\eta$ depends continuously on $R$. Moreover $R$ is continuous in $P$,
	hence the result.
\end{proof}

We require the following technical lemma in the proof of Theorem~\ref{thm:schnorr-stimm-dicho-bernoulli}. The first Markov-chain fact above is the standard finite-state ergodic input used in its proof.

\begin{lemma}
	\label{lem:ergodiconnormal}
	Let $(Q,A,q_I,\delta)$ be a finite-state automaton, $\mu$ be a positive Bernoulli measure and let $V_q(n,X)$ denote the number of times the state $q$ is visited upon running the automaton using the first $n$ bits of $X \in A^\omega$. 
	
	Then, for every $q \in Q$, there exists a real number $\pi_q \geq 0$ such that, for every $\mu$-normal sequence $X \in A^\omega$:
	\begin{itemize}
		\item either $V_q(n,X)$ is ultimately constant (i.e., $q$ is visited only finitely often during the run on~$X$)
		\item or, $	\lim\limits_{n \to \infty} \frac{V_q(n,X)}{n} = \pi_q$ (i.e., the state~$q$ is visited with asymptotic frequency $\pi_q$) and this second case can only happen when $\pi_q>0$.
	\end{itemize}

\end{lemma}

\begin{proof}
	On running an automaton upon a normal sequence, starting from any state, a strongly connected component must be reached in finitely many steps. Similar to \cite{SeillerS2020}, let us consider the Markov chain corresponding to the $\lvert Q \rvert \times \lvert Q \rvert$ stochastic matrix $\mathbf{P}$ where,
	\begin{align*}
		\mathbf{P}_{ij}=\sum\limits_{a \in A} \mu(a) \cdot 1_{\delta(i,a)=j}.	
	\end{align*}
	The proof follows using the Ergodic Theorem for Markov chains and the same steps in the proof of Lemma 4.5 from \cite{SchnorrS1972} by replacing the uniform measure with the positive Bernoulli measure induced by $\mu$. We note that this line of proof uses $\mu$-block normality, which is equivalent to $\mu$-normality (Theorem \ref{thm:mu-blocknormal-equivalence}).
\end{proof}

We now prove the Bernoulli version of the Schnorr-Stimm dichotomy.

\begin{proof}[Proof of Theorem~\ref{thm:schnorr-stimm-dicho-bernoulli}]
	In part $(i) \Rightarrow (ii)$ of our proof of~Theorem~\ref{thm:schnorr-stimm-bernoulli}, we showed that on a given state~$r$, either $r$ is not a betting state  ($\gamma$ is the constant~$1$ on this state) or it is and then the gambler  loses money exponentially fast in the number of times this state is visited, the exponent being $\alpha_r = \sum_{a \in A} \mu(a)\cdot  \log \gamma(r,a)$ which we proved to be negative. By Lemma~\ref{lem:ergodiconnormal}, betting states are either visited finitely often or with positive asymptotic density. If they are all visited finitely often, the capital stabilizes after the last bet is made. Otherwise, if betting states $r$ are visited with frequency $\pi_r$ and at least one $\pi_r$ is positive, then the gambler loses at an exponential rate, where the exponent is $\sum_r \pi_r \alpha_r$.
	
	In part $(ii) \Rightarrow (i)$, under the assumption of non-$\mu$-normality of~$X$, we built a gambler~$\G$ satisfying the following: for any fixed $\varepsilon$, there are arbitrarily large~$n$ and $k$ such that $
	\log \capital(\G,X[0,n]) \geq k \left( \mathrm{D}_{KL}(\nu || \mu) - \varepsilon \right)$. Here, $k$ is the number of visits to a state $q_u$. In the proof of Theorem \ref{thm:schnorr-stimm-bernoulli} part $(ii)$, this state $q_u$ is visited with frequency $\mu(u)$ for large enough~$n$. Hence the gambler has an `infinitely often' exponential rate of success with exponent $\mu(u) D_{KL}(\nu || \mu)$. 
\end{proof}

\section{The Agafonov theorem for PFAs}

We now want to prove the extension of Theorem~\ref{thm:agafonov-bernoulli} to probabilistic automata/selectors. A \emph{probabilistic finite automaton (PFA)} is a tuple $(Q,A,q_I,\delta)$ where $Q$ is a finite set of states, $A$ is a finite alphabet, $q_I$ is the initial state and $\delta: Q \times 	A \rightarrow \Delta(Q)$ is a probabilistic transition function, that is, $\Delta(Q)$ is the set of probability distributions over~$Q$. In this setting, we define inductively the random variables $\delta^*(q,w)$ by $\delta^*(q,\emptystr)=q$ and for $w \in A^*$ and $a \in A$, the event $\delta^*(q,w \cdot a) = q'$ is defined as the union
\[
\bigcup_{r \in Q} \left[ \delta^*(q,w)=r \wedge \delta_{|w|+1}(r,a)=q' \right]
\]
where $\{\delta_n(r,b) : n \in \mathbb{N}, r \in Q, b \in A\}$ is a family of independent random variables such that for all $(n,r,b)$, the distribution of $\delta_n(r,b)$ is $\delta(r,b)$.

Modulo this change of type of transition, probabilistic selectors as well as $\select$ are defined as before. This makes $\select(\S,X)$ a random variable for every given~$X \in A^{\leq \omega}$. We can now state Agafonov's theorem for PFAs, for which we will present two proofs. 

\begin{theorem}[Agafonov's theorem for PFA]\label{thm:agafonov-pfa}
	Let $X \in A^\omega$, and $\mu$ a Bernoulli measure over~$A$. The following are equivalent.
	\item[(i)] $X$ is $\mu$-normal. 
	\item[(ii)] For any deterministic selector $\S$ that selects a subsequence~$Y$ of~$X$, either~$Y$ is finite or $Y$ is $\mu$-normal. 
	\item[(iii)] For any probabilistic selector $\S$ that selects a subsequence~$Y$ of~$X$, almost surely, either~$Y$ is finite, or $Y$ is $\mu$-normal. 
\end{theorem}

\subsection{First proof: by reduction to the deterministic case}

In order to lift the deterministic Agafonov theorem for Bernoulli measures to the probabilistic case, we will need some preliminary lemmas about normality. \\

Given two alphabets $A$ and $B$, and given two sequences $X \in A^\omega$ and $Y \in B^\omega$, we denote by $X \otimes Y$ the sequence $Z$ over the alphabet $A \times B$ where $Z(n)=(X(n),Y(n))$. For two words $v$ and $w$ of the same length over $A$ and $B$ respectively, the product $v \otimes w$  is defined in the same way. Likewise, if $\mu$ and $\nu$ are Bernoulli measures over $A$ and $B$ respectively, $\mu \otimes \nu$ is the Bernoulli measure $\xi$ over $A \times B$ where $\xi((a,b))=\mu(a)\nu(b)$, for all $(a,b) \in A \times B$. Finally, if $Z$ is a sequence over a product alphabet $A \times B$, we denote by $\pi_0$ and $\pi_1$ its first and second projection respectively (in other words, $\pi_0(X \otimes Y)=X$ and $\pi_1(X \otimes Y)=Y$ for all $X,Y$). 

\begin{lemma}\label{lem:random-join}
	Let $A$ and $B$ be two alphabets and $\mu$, $\nu$ two Bernoulli measures over $A$ and $B$ respectively. If a sequence $X \in A^\omega$ is $\mu$-normal and a sequence $Y \in B^\omega$ is drawn at random according to $\nu$, then $\nu$-almost surely, $X \otimes Y$ is $\mu \otimes \nu$-normal. 
\end{lemma}

\begin{proof}
For this proof, it is easier to use block normality. Fix a word $w=u \otimes v$ of $(A \times B)^N$ for some~$N$. We split $X\otimes Y$ into blocks of length $N$ :
	\[
	X\otimes Y = (X_1\otimes Y_1) \cdot (X_2\otimes Y_2) \cdot  \ldots
	\]
	with $|X_i|=|Y_i|=N$, for all $i \geq 1$.\\
	
By block-normality of~$X$, 
\[
\# \{i \leq n \mid X_i = u\} = \mu(u) n + o(n)
\]	
For each~$i$ such that $X_i = u$, the probability of~$Y_i$ to be $v$ is $\mu(v)$, independently of all other~$Y_j$. Thus, by the law of large numbers, with probability~$1$ over~$Y$: 

\begin{align*}
\# \{i \leq n \mid X_i = u ~ \text{and}~ Y_i = v \} & = \left(\mu(u) n + o(n)\right) \nu(v) + o(n) \\
 & = \mu(u)\nu(v)n + o(n) \\
 & = (\mu \otimes \nu)(u \otimes v)n + o(n)
\end{align*}

And thus, with probability~$1$ over~$Y$, the block $u \otimes v$ appears in $X \otimes Y$ with its expected frequency. This being true for any $u \otimes v$, this shows that $X \otimes Y$ is block normal with probability~$1$ over~$Y$. 
\end{proof}

\begin{lemma}\label{lem:projection}
	If a sequence $Z$ is $\mu \otimes \nu$-normal over $A \times B$, then $\pi_0(Z)$ and $\pi_1(Z)$ are $\mu$-normal and $\nu$-normal respectively. 
\end{lemma}

\begin{proof}
	Suppose $Z \in (A \times B)^\omega$ is $\mu \otimes \nu$-normal. We only need to show that $\pi_0(Z)$ is $\mu$-normal, the proof of the $\nu$-normality of $\pi_1(Z)$ is the same by symmetry. Let $w \in A^n$. We have for every~$n$,
	\[
	\freq(w,\pi_0(Z)[0,n])  =  \sum_{w' \in B^n} \freq(w \otimes w',Z[0,n])
	\]
	which, by $(\mu \otimes \nu)$-normality of $Z$, implies 
	\[
	\freq(w,\pi_0(Z)[0,n])  =  \sum_{w' \in B^n} (\mu \otimes \nu)(w \otimes w') + o(1) =  \sum_{w' \in B^n} \mu(w)\nu(w') +o(1) = \mu(w) +o(1).
	\]
And thus, 	taking the limit, we get $\freq(w,\pi_0(Z)) = \mu(w)$, as wanted. 
\end{proof}

We now have all the necessary tools to prove Theorem~\ref{thm:agafonov-pfa}. 

\begin{proof}[Proof of Theorem~\ref{thm:agafonov-pfa}]
	Since deterministic selectors (for which we already have Agafonov's theorem) are a subset of the probabilistic ones, all that is left to prove is $(i) \Rightarrow (iii)$. 
	
	Let $\S=(Q,A,q_I,\delta,S)$ be a probabilistic selector. Recall that each transition $\delta(q,a)$ (where $q \in Q$ and $a \in A$) is some probability distribution over $Q$. Consider the set $\mathcal{T}$ of all functions $Q \times A \rightarrow Q$. We can put a distribution $\tau$ over~$\mathcal{T}$ such that if $t$ is chosen according to $\tau$, for every $(q,a)$, the marginal distribution of $t(q,a)$ is $\delta(q,a)$. An easy way to do this is to take $\tau= \bigotimes_{(q,a)\in Q \times A} \delta(q,a)$.
	
	This construction means that, for a fixed sequence~$X \in A^\omega$, an equivalent way to simulate the probabilistic run of $\S$ on~$X$ is, every time we are in a state~$q$ and read a letter~$a$, to pick $t$ at random according to $\tau$ and move to state $t(q,a)$. But $\mathcal{T}$ is a finite set and $\tau$ a Bernoulli measure over it. This Bernoulli measure might not be positive. If it is not, we simply remove from $\mathcal{T}$ all functions whose $\tau$-probability is~$0$. Reformulating slightly, yet another equivalent way to simulate the run of~$\S$ over~$X$ is to do the following:
	\begin{itemize}
		\item[1.] First choose $T \in \mathcal{T}^\omega$ at random with respect to the Bernoulli measure $\tau$. 
		\item[2.] Then run on the sequence $X \otimes T$ the deterministic selector $\hat{\mathcal{S}}$ whose set of states is~$Q$ and transition $\hat{\delta}$ is defined by $\hat{\delta}(q,(a,t))=t(q,a)$. 
	\end{itemize}
	
	Now, the two following random variables have the same distribution: 
	\begin{itemize}
		\item The subsequence $Y$ of $X$ selected by~$\S$.
		\item The sequence $\pi_0(\hat{Y})$, where $\hat{Y}$ is the subsequence of $X \otimes T$ selected by $\hat{\S}$ when $T$ is chosen randomly according to $\tau$. 
	\end{itemize}
	
	Since $X$ is $\mu$-normal, by Lemma~\ref{lem:random-join}, $X \otimes T$ is $\mu \otimes \tau$-normal $\tau$-almost surely. Thus, by Agafonov's theorem for deterministic selectors and Bernoulli measures (Theorem~\ref{thm:agafonov-bernoulli}), almost surely, the subsequence $\hat{Y}$ selected by $\hat{\S}$ is $\mu \otimes \tau$-normal. Finally, by Lemma~\ref{lem:projection}, this implies that almost surely, the subsequence $\pi_0(\hat{Y})$ of $X$ is $\mu$-normal. This concludes the proof. 
\end{proof}

\subsection{Second proof: direct analysis}

We now provide a second proof of Theorem~\ref{thm:agafonov-pfa},
along the lines of Carton's techniques~\cite{Carton2020},
with a different presentation based on the notion of \emph{balance-inducing input words}. Let $\S$ be a probabilistic selector with set of states~$Q$, selecting states~$S \subseteq Q$ and probabilistic transition function~$\delta$. 

\subsubsection{The strongly connected case} \label{subsubsec:scpc}  

For the time being, we assume that 
\[
\text{$\S$ is strongly connected},
\]
i.e., for all states $q,r \in Q$,
there is a computation in $\A$ from $q$ to $r$ (i.e., for some~$w$, the event ``$\delta^*(q,w)=r$" has positive probability).
We also assume that $\S$ has at least one selecting state, i.e., $S \neq \emptyset$.

%
%

For now, suppose that instead of having a fixed~$\mu$-normal~$X$, we draw~$X \in A^\omega$ at random according to the probability measure~$\mu$. This turns the sequence $q_1, q_2, \ldots$ of states visited into a Markov chain~$\M$. By the ergodic theorem for Markov chains, with probability~$1$, during this process, each state $q \in Q$ is visited with asymptotic frequency $\pi_q >0$. Moreover, each time we land on a state~$q \in S$, a letter will be output in the selected subsequence at the next step. Thus, with probability 1, after having read~$n$ letters of input~$X$, we will asymptotically obtain $\lambda n +o(n)$ letters of the output~$Y$, where $\lambda  = \sum_{q \in S} \pi_q>0$. And this is true no matter what the starting state is.   

Furthermore, observe that a decision to select a letter~$X(n)$ in the subsequence~$Y$ only depends on the values of $X(0), \ldots X(n-1)$ and the internal randomness of the PFA, which are independent of the random variable~$X(n)$. This means in particular that each value of $Y(k)$ is independent of $Y(0), \ldots Y(k-1)$ and has distribution~$\mu$. Therefore, the subsequence~$Y$ is infinite with probability~$1$ and distributed according to~$\mu$. A fortiori, $Y$ will be $\mu$-normal with probability~$1$.\\

The process described above uses two sources of randomness: the sequence~$X$ and the internal randomness, call it $R$, of the PFA. What we have argued is that the set $\mathcal{B}$ of pairs $(X,R)$ such that, for any starting state~$q_0$: 
\begin{itemize}
\item the output~$Y$ is $\mu$-normal and 
\item $(\lambda+o(1)) n$ letters are output when $n$ bits of~$X$ are read
\end{itemize}
has measure~$1$. \\

By Fubini's theorem, for almost all~$X$, the projection $\mathcal{B}_X$ of the set~$\mathcal{B}$ over its first coordinate~$X$ has measure~$1$. Rephrasing this result, unfolding its many logical quantifiers: For almost all~$X$, for every word~$w \in A^*$, for every~$\varepsilon>0$, there exists an~$N$ such that,  for any $n \geq N$, for any starting state~$q_0$, with probability at least $(1-\varepsilon)$ (over~$R$), at least $(1-\varepsilon)\lambda n$ bits are output after reading~$n$ bits of~$X$, and the word~$w$ appears with frequency at least $(1-\varepsilon)\mu(w)$ in this output.  

The integer~$N$ depends on~$X$ (as well as $w$ and $\varepsilon$) but if we fix $\delta>0$, we obtain as a corollary of the above that there is some~$N=N(w, \varepsilon, \delta)$ such that for a measure at least $(1-\delta)$ of $X's$, for any starting state~$q_0$, with probability at least $(1-\varepsilon)$ (over~$R$), at least $(1-\varepsilon)\lambda N$ bits are output after reading the first~$N$ bits of~$X$, and the word~$w$ appears at least $(1-\varepsilon)\mu(w)L$ times, where~$L$ is the length of this output.  

In this last analysis, only the first~$N$ bits of~$X$ come into play. So what we have really proven is the following. 

\begin{lemma}
Let $\varepsilon, \delta > 0$ and $w \in A^*$ be fixed. There exists an~$N=N(w, \varepsilon, \delta)$ such that, if we call $G(w, \varepsilon, \delta)$ the set of strings $x \in A^N$ with the following properties:
\begin{itemize}
\item for any starting state~$q_0$, on input~$x$, with probability at least $(1-\varepsilon)$ (over the internal randomness of the automaton~$\mathcal{A}$), the number of output letters belongs to $[(1-\varepsilon)\lambda N, (1+\varepsilon)\lambda N]$,
\item the word~$w$ appears at least $(1-2\varepsilon)\mu(w)\lambda N$ times in this output
\end{itemize}
we then have that $G(w, \varepsilon, \delta)$ has $\mu$-measure at least $(1-\delta)$ (i.e., $\sum_{u \in G(w, \varepsilon, \delta)} \mu(u) \geq (1-\delta)$). 
\end{lemma}

Now, let us go back to the case where~$X$ is not random but is a fixed $\mu$-normal sequence and~$w$ a fixed word. We want to argue that on input~$X$, almost surely, the automaton~$\mathcal{A}$ outputs an infinite sequence~$Y$ in which~$w$ appears with frequency~$\mu(w)$. Fix $\varepsilon>0$ and $\delta>0$. Let us split~$X$ into blocks of length~$N=N(w, \varepsilon, \delta)$, say $X=x_1 x_2 x_3 \ldots$ with $|x_i|=N$ for all~$i$. Let us call $u_i$ the word consisting of the letters output by $\mathcal{A}$ while reading the block~$x_i$ (so $u_i$ is a random variable). 

By $\mu$-normality and the previous lemma, an asymptotic fraction at least $(1-\delta)$ of the $x_i$'s belong to $G(w, \varepsilon, \delta)$. In other words, most (a fraction at least $1-\delta$) of blocks of~$X$ are `good', in the sense that they cause, with high probability (at least $1-\varepsilon$), the word~$u_i$ to have length in $[(1-\varepsilon)\lambda N, (1+\varepsilon)\lambda N]$ with $w$ appearing at least $(1-2\varepsilon)\mu(w)\lambda N$ in $u_i$.

In order to apply this to the infinite run of the automaton~$\mathcal{A}$ on input~$X$, we need the following elementary lemma.

\begin{lemma}
Let $(R_n)_{i \in \mathbb{N}}$ be a sequence of random variables over some discrete space $\mathcal{R}$. Let $(E_n)_{i \in \mathbb{N}}$ be a sequence of events and $\alpha >0$ such that:
\begin{itemize}
\item For all~$n$, $E_n$ is $(R_1, \ldots, R_n)$-measurable. (That is, whether $E_n$ holds or not is fully determined by the values of $(R_1, \ldots, R_n)$). 
\item For every $(r_1, \ldots, r_{n-1}) \in \mathcal{R}^{n-1}$, we have that $\mathbb{P}(E_n \mid R_1=r_1, \ldots, R_{n-1}=r_{n-1}) \geq \alpha$ (provided this expression is well-defined, i.e., when the condition has positive probability). 
\end{itemize}

For all~$n$, let $S_n = \sum_{i=1}^n \mathbf{1}_{E_i}$. Under the above hypotheses we have, almost surely, $\liminf_{n \rightarrow +\infty} \frac{S_n}{n} \geq \alpha$. 
\end{lemma}

\begin{proof}
The hypotheses directly imply that the sequence of random variables $S'_n = S_n - \alpha n$ is a submartingale with respect to the sequence $(R_n)_{n \in \mathbb{N}}$. Moreover, by definition, we have $|S'_{n+1}-S'_n| \leq 1$ for all~$n$. Thus, we can apply Azuma's inequality for submartingales (see for example~\cite{AlonS2000}) and get
\[
\mathbb{P}(S'_n \leq - n^{3/4}) \leq e^{-2 \sqrt{n}}
\]
Since the series $\sum_n e^{-2 \sqrt{n}}$ is convergent, by the Borel-Cantelli lemma, we must have, with probability~$1$, $S'_n > - n^{3/4}$ for almost all~$n$. In other words, $S_n > \alpha n - n^{3/4}$ for almost all~$n$, and the lemma follows. 
\end{proof}

We will use this lemma by taking for $R_n$ the $N$-uple of states visited by the automaton~$\A$ on input~$X$ while reading the block $x_n$ and for $E_n$ the event: ``either $x_n \notin G(w, \varepsilon, \delta)$, or the length of ~$u_n$ is in the interval $[(1-\varepsilon)\lambda N, (1+\varepsilon)\lambda N]$, and the word~$w$ appears at least $(1-2\varepsilon)\mu(w)\lambda N$ times in $u_n$". The event~$E_n$ is indeed $(R_1, \ldots, R_n)$-measurable since knowing $(R_1, \ldots, R_n)$ allows one to reconstruct the entire run of the automaton on the prefix $x_1x_2 \ldots x_n$ of~$X$. Moreover, by definition of $G(w, \varepsilon, \delta)$, we have that for all possible values $(r_1, \ldots, r_{n-1})$ of $(R_1, \ldots, R_n)$, $\mathbb{P}(E_n \mid R_1=r_1, \ldots, R_{n-1}=r_{n-1}) \geq 1-\varepsilon$. 

Applying the above lemma, if we consider the run of the automaton $\mathcal{A}$ on~$X$ after reading $x_1 \ldots x_n$:
\begin{itemize}
\item The number of $x_i'$'s that are not in $G(w, \varepsilon, \delta)$ is at most $(\delta+o(1))n$, by normality of~$X$.
\item By the above lemma, with probability~$1$, a fraction at least $1-\varepsilon-o(1)$ of events $\{E_i \, : \, i \leq n\}$ happen and thus, for a fraction at least $(1-\delta-\varepsilon-o(1))$ of indices $i \leq n$, the word $u_i$ has a length in $[(1-\varepsilon)\lambda N, (1+\varepsilon)\lambda N]$ and the word~$w$ appears at least $(1-2\varepsilon)\mu(w)\lambda N$ times in $u_i$ (for other~$i$, we will only use that the length of $u_i$ is at most~$N$)
\end{itemize}

Putting this all together, with probability~$1$, we have
\[
|u_1 \ldots u_n| \leq (1-\delta-\varepsilon)(1+\varepsilon)\lambda Nn + (\delta+\varepsilon)Nn +o(n)
\]
and 
\[
\nbocc(w,u_1 \ldots u_n) \geq (1-\delta-\varepsilon)(1-2\varepsilon)  \mu(w) \lambda N n +o(n)
\]
But since $\delta$ and $\varepsilon$ can be taken arbitrarily small, this means that with probability~$1$, 
\[
\freq(w,u_1 \ldots u_n) \geq \mu(w)  +o(n)
\]
And thus, 
\[
\freq^-(w,Y) \geq \mu(w)
\]
Note that~$w$ was arbitrary hence we have established this for any word~$w$. Now, fix again a word~$w$ and observe that 
\[
\freq(w,u_1 \ldots u_n) = 1- \sum_{|w'|=|w| \atop w' \not= w}  \freq(w',u_1 \ldots u_n) 
\]
Which implies 
\[
\freq^+(w,Y) \leq 1 -  \sum_{|w'|=|w| \atop w' \not= w}  \freq^-(w',Y) \leq 1 - \sum_{|w'|=|w| \atop w' \not= w}  \mu(w') = \mu(w)
\]
Thus the inf-frequency and sup-frequency of~$w$ inside~$Y$ coincide, and we have finally obtained that, with probability~$1$:
\[
\freq(w,Y) = \mu(w)
\]
which finishes the proof of the strongly connected case.

\subsubsection{The general case}

Now, we get to the general case where the automaton may have a non-empty set of transient states~$T$.

By definition, from every state in $T$ there is a word and a computation
on this word which leaves $T$ and enters a bottom strongly connected component (i.e., a strongly connected component with no outgoing edge to other strongly connected components),
from which there is no way back to $T$.
We show
\[
\PP_{\A,z}(\exists n , Q_n \not \in T) = 1\enspace.
\]
From every state $q\in T$,
there is at least one word of length $\leq |T|$ which 
exits $T$ 
with nonzero probability,
bounded from below by ${p_m}^{|T|}$ where $p_m$
is the smallest non-zero probability appearing in $p$.
We can concatenate $|T|$ such words in order to obtain a single word
$u$ which guarantees probability $\geq {p_m}^{|T|^2}$ to leave $T$,
 \emph{whatever the initial state}. 
Since the input word~$z$ is normal, $u$ appears infinitely often in $z$,
thus $T$ is eventually left almost-surely.

If the bottom strongly connected component reached after leaving $T$ contains no selecting state, then no further letters are selected after this point, and hence the selected output is finite. Otherwise, this bottom component contains at least one selecting state, and the suffix of the computation restricted to this component falls under the strongly connected case treated above. Since normality is a tail property, the corresponding suffix of the input word is normal; therefore, by the strongly connected case, almost surely the corresponding suffix of the selected output is normal. It follows again by tail-invariance of normality that the whole selected output is normal.

\section{The Schnorr-Stimm theorem for probabilistic gamblers}

In this last section, we establish the Schnorr-Stimm theorem for probabilistic finite-state automata. The definition of probabilistic ($\mu$-)gambler is the same as the definition of deterministic ($\mu$-)gambler except that it is based on PFAs instead of DFAs. In this setting, the quantities $\capital(\G,w)$ become random variables.

\begin{theorem}[Schnorr-Stimm theorem for PFAs]
	For $X \in A^\omega$ and $\mu$ a Bernoulli measure, the following are equivalent.
	\item[(i)] $X$ is $\mu$-normal. 
	\item[(ii)] Any probabilistic $\mu$-gambler $\G$ betting on~$X$ loses almost surely.   
\end{theorem}

\begin{proof}
	We already have $(ii) \Rightarrow (i)$ by Theorem~\ref{thm:schnorr-stimm-bernoulli} (deterministic gamblers being a subset of probabilistic ones). 
	
	For $(i) \Rightarrow (ii)$, the idea is similar to the proof of Theorem~\ref{thm:agafonov-pfa}. To simulate the run of a probabilistic $\mu$-gambler $\G=(Q,A,q_I,\delta,\gamma)$ on a sequence~$X$, one can equivalently run on the sequence $X \otimes T$ (where $T \in \mathcal{T}^\omega$ is $\tau$-random sequence)\footnote{$\mathcal{T}$ and $\tau$ are defined in the same way as in Theorem~\ref{thm:agafonov-pfa}.} the deterministic gambler $\hat{\G}=(Q,A,q_I,\hat{\delta},\hat{\gamma})$ where again $\hat{\delta}(q,(a,t))=t(q,a)$ for all~$(q,a,t) \in Q \times A \times \mathcal{T}$ and $\hat{\gamma}(q,(a,t)) = \gamma(q,a)$ (indeed the bet placed by a gambler at a given stage does not take into account which state will be reached next). Note that the fairness condition is respected since for every~$q$,
	\begin{eqnarray*}
		\sum_{(a,t) \in A \times \mathcal{T}} (\mu \otimes \tau)(a,t) \cdot \hat{\gamma}(q,(a,t))  & = & \sum_{a \in A, t \in \mathcal{T}} \mu(a) \tau(t) \gamma(q,a) \\
		& = & \sum_{t \in \mathcal{T}} \tau(t) \sum_{a \in A} \mu(a) \gamma(q,a)\\
		& = & \sum_{t \in \mathcal{T}} \tau(t) \\
		& = & 1
	\end{eqnarray*}
	(the second-to-last equality holds by fairness condition on~$\gamma$ and the last one because $\tau$ is a distribution).
	
	In this way, the two following random variables will have the same distribution:
	\begin{itemize}
		\item $\capital(\G,X[0,n])$.
		\item $\capital(\hat{\G},(X \otimes T)[0,n])$ where $T \in \mathcal{T}^\omega$ is chosen at random according to $\tau$. 
	\end{itemize} 
	
	Since $X$ is $\mu$-normal, by Lemma~\ref{lem:random-join} again, $X \otimes T$ is $\mu \otimes \tau$-normal $\tau$-almost surely. Thus, by Theorem~\ref{thm:schnorr-stimm-bernoulli}, for $\tau$-almost all~$T$, $\capital(\hat{\G},(X \otimes T)[0,n])$ is bounded by a constant~$C$ independent on~$n$. By the equivalence of the two random variables above, this means that almost surely, there exists a constant~$C$ such that $\capital(\G,X[0,n])<C$ for all~$n$. In other words, almost surely, $\G$ loses on the sequence~$X$. 
\end{proof}

We remark that the dichotomy theorem for Bernoulli measures yields the following dichotomy for probabilistic $\mu$-gamblers.

\begin{theorem}[Schnorr-Stimm dichotomy theorem for PFAs]
	\label{thm:probabilisticdichotomy}
	Let $X$ be an infinite sequence in $A^\omega$ and $\mu$ a Bernoulli measure. \\
	(i) If $X$ is $\mu$-normal and $\G$ is a $\mu$-gambler, then almost surely the capital of $\G$ throughout the game either is ultimately constant or decreases at an exponential rate.\\
	(ii) If $X$ is not $\mu$-normal, then there exists a $\mu$-gambler~$\G$ which wins against~$X$ at an `infinitely often' exponential rate almost surely.
\end{theorem}

\section{Conclusion}
The main contributions of this paper are a generalization
of the Agafonov theorem for PFA with arbitrary transition probabilities which settles the open question posed by L\'echine et al.~\cite{LechineSS2024}, and an extension of Schnorr-Stimm theorem to probabilistic gamblers.  

While we proved the probabilistic Agafonov theorem (Theorem~\ref{thm:agafonov-pfa}) by reduction to the deterministic setting, it is also possible to follow with a more direct approach (i.e., without appealing to the Seiller-Simonsen result), similar to the one followed by Carton to generalize Agafonov's theorem for DFA~\cite{Carton2020}. This however makes the argument somewhat more complicated. 

An interesting direction for future research is to explore whether the `uselessness of randomness' also holds for pushdown automata. These are a more powerful model of computation and indeed some normal sequences can be predicted by pushdown automata (some in a rather dramatic way, as proven by Carton and Perifel~\cite{CartonP2024}\footnote{The result proven by Carton and Perifel considers a slightly different paradigm, namely compression (which we did not discuss in this paper) instead of prediction.}). We can for example ask: If some probabilistic pushdown selector selects a biased subsequence from a sequence~$X$, does there necessarily exist a deterministic pushdown selector which also selects a biased subsequence? Similarly, if some probabilistic pushdown gambler wins against a sequence~$X$, does there necessarily exist a deterministic pushdown gambler which wins against that same sequence~$X$?

A related question concerns the speed of success in the gambling model. In the case of finite-state automata, Schnorr and Stimm proved that either a sequence~$X$ cannot be predicted or some gambler wins on it at an exponential rate. This dichotomy no longer holds for pushdown automata, but one may ask the following question: If some probabilistic pushdown gambler wins at an exponential rate against a sequence~$X$, does there necessarily exist a pushdown gambler which wins against that sequence~$X$ at an exponential rate?

\bibliographystyle{alphaurl}
\bibliography{biblio-agafonov-pfa}

\end{document}